\documentclass[preprint]{elsarticle}

\usepackage{lineno,hyperref}
\usepackage{graphicx,amsmath,amssymb,amsthm,bm}
\modulolinenumbers[5]

\journal{}

\newtheorem{definition}{Definition}
\newtheorem{proposition}{Proposition}

\newtheorem{theorem}{Theorem}

\begin{document}

\begin{frontmatter}

\title{On the implementation of zero-determinant strategies in repeated games}

\author{Masahiko Ueda\corref{mycorrespondingauthor}}
\address{Graduate School of Sciences and Technology for Innovation, Yamaguchi University, Yamaguchi 753-8511, Japan}
\ead{m.ueda@yamaguchi-u.ac.jp}

\begin{abstract}
Zero-determinant strategies are a class of strategies in repeated games which unilaterally control payoffs.
Zero-determinant strategies have attracted much attention in studies of social dilemma, particularly in the context of evolution of cooperation.
So far, not only general properties of zero-determinant strategies have been investigated, but zero-determinant strategies have been applied to control in the fields of information and communications technology and analysis of imitation.
Here, we further deepen our understanding on general mathematical properties of zero-determinant strategies.
We first prove that zero-determinant strategies, if exist, can be implemented by some one-dimensional transition probability.
Next, we prove that, if a two-player game has a non-trivial potential function, a zero-determinant strategy exists in its repeated version.
These results assist us to implement zero-determinant strategies in broader situations.
\end{abstract}


\end{frontmatter}


\section{Introduction}
\label{sec:intro}
Zero-determinant (ZD) strategies are a class of strategies in repeated games which unilaterally enforce linear relations between payoffs.
They were originally discovered in the infinitely repeated prisoner's dilemma games \cite{PreDys2012}.
One of the examples is the equalizer strategy, which unilaterally controls the payoff of the opponent.
Another example is the extortionate strategy, which cannot be beaten.
There is also a ZD strategy called the generous ZD strategy, which cannot win but promotes mutual cooperation \cite{StePlo2013}.
After the discovery of ZD strategies, properties of ZD strategies have extensively been investigated in the context of evolution of cooperation \cite{HNS2013,AdaHin2013,StePlo2013,HNT2013,Aki2016,CWF2022,CheFu2023}, including evolution on structured populations \cite{SzoPer2014,SzoPer2014b,CGCF2024}, together with human experiments of ZD strategies \cite{HRM2016,WZLZX2016,BecMil2019}.
Furthermore, beyond the prisoner's dilemma game, ZD strategies were also discovered in multi-player social dilemma games \cite{HWTN2014,PHRT2015} and a continuous donation game \cite{McAHau2016}.
These results have suggested that ZD strategies can be used to control payoffs in a broader class of social dilemma situations.

At the same time as such investigation of ZD strategies in social dilemma situations, general properties of ZD strategies have also been studied.
For example, for games with a discount factor, it has been recognized that the range of possible ZD strategies shrinks compared with no discounting case \cite{HTS2015,McAHau2016,IchMas2018,GovCao2020}.
It was also found that effect of imperfect monitoring about the previous actions reduces the range of their existence, compared to perfect monitoring case \cite{HRZ2015,UedTan2020}.
Difference between these two effects in the prisoner's dilemma game was detailly investigated in Ref. \cite{MamIch2020}.
Existence of ZD strategies in asynchronous-update social dilemma games was also shown in Ref. \cite{McAHau2017}.
Furthermore, ZD strategies in general multi-player multi-action games were introduced in Ref. \cite{HDND2016}, and their mathematical properties in situations where several players use ZD strategies were also studied \cite{UedTan2020}, which found that possible ZD strategies are constrained by consistency of linear payoff relations.
A necessary and sufficient condition for the existence of ZD strategies was recently specified \cite{Ued2022b}.
Moreover, there are several attempts to extend the ability of payoff control of ZD strategies so as to control the moments of payoffs \cite{Ued2021}, time correlation functions of payoffs \cite{Ued2021b}, and conditional expectations of payoffs \cite{Ued2022c}.

Apart from these theoretical studies on general properties of ZD strategies, there are many applications of ZD strategies.
In the field of information and communications technology, ZD strategies have been applied in resource sharing in wireless networks \cite{DKL2014,ZNSet2016}, smart grid systems \cite{FHAet2015}, mining in blockchain \cite{TLYet2019,SWHet2020}, crowdsourcing \cite{TLCet2019,HWMet2019}, mobile crowdsensing \cite{HWCet2019}, the Internet of Things \cite{WSHet2019}, cloud computing \cite{ZTYet2020}, federated learning \cite{HWXet2021}, and data trading \cite{WSHet2021}, in order to unilaterally realize some favorable states.
Another field of application is analysis of imitation in behavioral science \cite{Ued2022,Ued2023}.
Originally, in the prisoner's dilemma game, it was pointed out that the Tit-for-Tat (TFT) strategy \cite{RCO1965,AxeHam1981}, which imitates the previous action of the opponent, is a ZD strategy which unilaterally equalizes the payoffs of two players \cite{PreDys2012}.
Because it has been known that TFT is unbeatable in two-player symmetric potential games \cite{DOS2014}, relation between unbeatability and ZD strategies was focused on, and it was indeed shown that unbeatable TFT is a ZD strategy \cite{Ued2022}.
Furthermore, relation between the existence of unbeatable imitation \cite{DOS2012b} and that of unbeatable ZD strategies has also been shown in two-player symmetric games \cite{Ued2022b} and in multi-player totally symmetric games \cite{Ued2023}.
These application studies show that the concept of ZD strategies can be useful as methods to control multi-agent systems and to analyze human behavior.

Although the concept of ZD strategy has been widely used, its general mathematical properties have not yet been sufficiently clear.
Clarifying its properties will enable us to implement ZD strategies in broader situations which were not in the scope of application of ZD strategies \cite{XWPW2023,CDPP2024}.
In this paper, we further clarify mathematical properties of ZD strategies.
Particularly, we first show that ZD strategies, if exist, can be realized by one-dimensional transition probability.
This result enables us to implement ZD strategies even if an action set is restricted to one-dimension space.
Second, we prove that ZD strategies controlling payoffs exist in two-player non-trivial potential games.
This result can be regarded as an extension of the result in two-player symmetric potential games \cite{Ued2022} to asymmetric cases.

This paper is organized as follows.
In Section \ref{sec:preliminaries}, we introduce basic concepts such as repeated games, zero-determinant strategies, and potential games.
In Section \ref{sec:physicalZDS}, we prove our first theorem about implementation of ZD strategies by one-dimensional transition probability.
In Section \ref{sec:ZDS_potential}, we provide our second theorem about the existence of ZD strategies controlling payoffs in two-player non-trivial potential games.
This result is also demonstrated by using an asymmetric prisoner's dilemma game in this section.
In Section \ref{sec:discussion}, we discuss relations between our results and previous results.
Section \ref{sec:conclusion} is devoted to concluding remarks.

\section{Preliminaries}
\label{sec:preliminaries}

\subsection{Repeated games}
A game in strategic form is defined as $G:=\left( \mathcal{N}, \left\{ A_j \right\}_{j\in \mathcal{N}}, \left\{ s_j \right\}_{j\in \mathcal{N}} \right)$, where $\mathcal{N}$ is the set of players, $A_j$ is the set of actions of player $j\in \mathcal{N}$, and $s_j: \prod_{j\in \mathcal{N}}A_j \rightarrow \mathbb{R}$ is the payoff function of player $j\in \mathcal{N}$ \cite{FudTir1991,OsbRub1994}.
We consider the situation that $\mathcal{N} := \left\{ 1, \cdots, N \right\}$ and $A_j:=\left\{ b_j^{(1)}, \cdots, b_j^{(L_j)} \right\}$ for all $j\in \mathcal{N}$, where $L_j \in \mathbb{N}$ is the number of action of player $j$.
For simplicity, we write $\mathcal{A}:=\prod_{j=1}^N A_j$ and $\bm{a}:=\left( a_1, \cdots,  a_N \right)\in \mathcal{A}$, and call $\bm{a}$ an action profile.
Below we also use notation $a_{-j}:= \left( a_1, \cdots, a_{j-1}, a_{j+1}, \cdots, a_N \right)$ and $A_{-j}:=\prod_{k\neq j} A_k$.
When we want to emphasize $a_j$ in $\bm{a}$, we write $\bm{a}=\left( a_j, a_{-j} \right)$.

When $G$ is infinitely repeated and players choose actions at each round according to all histories of actions, such game is called an (infinitely) repeated game.
We write an action of player $j$ at round $t\geq 1$ as $a_j^{(t)}$.
Below $\Delta_L$ denotes a probability $L$-simplex.
As strategies in repeated games, we consider behavior strategies $\mathcal{T}_j := \left\{ T^{(t)}_j \right\}_{t=1}^\infty$ $(\forall j\in \mathcal{N})$, where $T^{(t)}_j: \mathcal{A}^{t-1} \to \Delta_{L_j}$ with $t\geq 2$ is the conditional probability at $t$-th round and $T^{(1)}_j \in \Delta_{L_j}$ is the initial probability.
Similarly to the original work \cite{PreDys2012}, in this paper, we only consider the case that payoffs in the infinitely repeated game are given by the time-averaged payoffs (that is, no discounting)
\begin{align}
 \mathcal{S}_j &:= \lim_{T\rightarrow \infty} \frac{1}{T} \sum_{t=1}^T \mathbb{E} \left[ s_j\left( \bm{a}^{(t)} \right) \right] \quad (\forall j\in \mathcal{N}),
 \label{eq:payoff_repeated}
\end{align}
where $\mathbb{E}[\cdots]$ is the expectation of the quantity $\cdots$ with respect to strategies of all players.

We write the joint probability distribution of action profiles from first round to $t$-th round as $P\left( \bm{a}^{(t)}, \cdots, \bm{a}^{(1)} \right)$.
We also define the marginal probability distribution of an action profile at $t$-th round obtained from the joint probability distribution by
\begin{align}
 P_t \left( \bm{a}^{(t)} \right) &:= \sum_{\bm{a}^{(t-1)}} \cdots \sum_{\bm{a}^{(1)}} P\left( \bm{a}^{(t)}, \cdots, \bm{a}^{(1)} \right) \quad \left( \forall \bm{a}^{(t)} \in \mathcal{A} \right).
\end{align}
If the limit probability distribution
\begin{align}
 P^* \left( \bm{a} \right) &:= \lim_{T\rightarrow \infty} \frac{1}{T} \sum_{t=1}^T P_{t} \left( \bm{a} \right) \quad (\forall \bm{a} \in \mathcal{A})
 \label{eq:P_star}
\end{align}
exists, we write the expectation of the quantity $D:\mathcal{A}\rightarrow \mathbb{R}$ with respect to $P^*$ by $\left\langle D \right\rangle^{*}$.
For such case, the payoffs (\ref{eq:payoff_repeated}) are described as $\mathcal{S}_k=\left\langle s_k \right\rangle^{*}$ for all $k \in \mathcal{N}$.

\subsection{Zero-determinant strategies}
Here, we define a \emph{time-independent memory-one strategy} of player $j\in \mathcal{N}$ as a strategy such that
\begin{align}
 T^{(t)}_j\left( a_j^{(t)} | \bm{a}^{(t-1)}, \cdots, \bm{a}^{(1)} \right) &= T_j \left( a_j^{(t)} | \bm{a}^{(t-1)} \right) \quad (\forall t\geq 2)
\end{align}
for all $a_j^{(t)}$, $\bm{a}^{(t-1)}$, $\cdots$, $\bm{a}^{(1)}$, where $T_j: \mathcal{A} \to \Delta_{L_j}$.
That is, a player using time-independent memory-one strategy chooses an action in each round only by using information of a previous action profile.
Furthermore, for time-independent memory-one strategies $T_j$ of player $j$, we introduce the \emph{Press-Dyson vectors} \cite{Aki2016,McAHau2016,UedTan2020}
\begin{align}
 \hat{T}_j\left( a_j | \bm{a}^{\prime} \right) &:= T_j\left( a_j | \bm{a}^{\prime} \right) -  \delta_{a_j, a^{\prime}_j} \quad \left( \forall a_j \in A_j, \forall \bm{a}^{\prime} \in \mathcal{A} \right),
 \label{eq:PD}
\end{align}
where $\delta_{a, a^\prime}$ is the Kronecker delta.
Note that the Press-Dyson vectors describe the difference between the time-independent memory-one strategy $T_j$ and the strategy which repeats one's previous action (the Repeat strategy \cite{Aki2016}) $\delta_{a_j, a^{\prime}_j}$.

We remark that, because $T_j$ is probability, Press-Dyson vectors generally satisfy
\begin{align}
 \sum_{a_j} \hat{T}_j \left( a_j | \bm{a}^\prime \right) &= 0 \quad \left( \forall \bm{a}^\prime \right)
 \label{eq:PD_normalized}
\end{align}
and
\begin{align}
  & \left\{
  \begin{array}{ll}
    -1 \leq \hat{T}_j \left( a_j | \bm{a}^\prime \right) \leq 0 & \left(a_j = a^\prime_j \right) \\
    0 \leq \hat{T}_j \left( a_j | \bm{a}^\prime \right) \leq 1 & \left(a_j \neq a^\prime_j \right).
  \end{array}
  \right.
  \label{eq:PD_range}
\end{align}
Now we introduce zero-determinant strategies in repeated games \cite{PreDys2012,McAHau2016,Ued2022c}.
\begin{definition}
\label{def:ZDS}
A time-independent memory-one strategy of player $j$ is an (extended) \emph{zero-determinant (ZD) strategy} controlling the quantity $B:\mathcal{A}\rightarrow \mathbb{R}$ when its Press-Dyson vectors can be written in the form
\begin{align}
 \sum_{a_j} c_{a_j} \hat{T}_j\left( a_j | \bm{a}^{\prime} \right) &= B \left( \bm{a}^{\prime} \right) \quad \left( \forall \bm{a}^{\prime} \in \mathcal{A} \right)
 \label{eq:ZDS}
\end{align}
with some nontrivial coefficients $\left\{ c_{a_j} \right\}$ (that is, not $c_{b_1}=\cdots=c_{b_{L_j}}=\mathrm{const.}$) and $B$ is not identically zero.
\end{definition}

Counterintuitively, the ZD strategy unilaterally enforces
\begin{align}
 \left\langle B \right\rangle^* &= 0,
 \label{eq:linear}
\end{align}
that is, $\left\langle B \right\rangle^*=0$ holds independently of the strategies of other players \cite{PreDys2012,McAHau2016,Ued2022c}.
Originally, ZD strategies controlling $B(\bm{a}) = \sum_{j=1}^N \alpha_j s_j(\bm{a}) + \alpha_0$ with some coefficients $\left\{ \alpha_j \right\}_{j=0}^N$ have been investigated.

An existence condition of ZD strategies is given by the following proposition:
\begin{proposition}[\cite{Ued2022b}]
\label{prop:existence}
A ZD strategy of player $j$ controlling $B$ exists if and only if there exist two different actions $\overline{a}_j, \underline{a}_j \in A_j$ of player $j$ such that
\begin{align}
 B \left( \overline{a}_j, a_{-j} \right) &\geq 0 \quad \left( \forall a_{-j} \right) \nonumber \\
 B \left( \underline{a}_j, a_{-j} \right) &\leq 0 \quad \left( \forall a_{-j} \right),
 \label{eq:condition_exsitence}
\end{align}
and $B$ is not identically zero.
\end{proposition}

\subsection{Potential games}
The concept of potential games was originally introduced in Ref. \cite{MonSha1996} inspired by physics.
\begin{definition}
\label{def:potential}
A game $G$ is an \emph{(exact) potential game} when there exists a common function $\Phi(\bm{a})$ satisfying
\begin{align}
 s_j(a_j, a_{-j}) - s_j(a^\prime_j, a_{-j}) &= \Phi(a_j, a_{-j}) - \Phi(a^\prime_j, a_{-j}) \quad (\forall a_j, \forall a^\prime_j, \forall a_{-j})
 \label{eq:potential}
\end{align}
for all player $j$.
\end{definition}
The function $\Phi$ is called a potential function.
It should be noted that the Nash equilibrium $\bm{a}^*$ is defined by the condition
\begin{align}
 s_j\left( a_j^*, a_{-j}^* \right) &\geq s_j \left( a_j, a_{-j}^* \right) \quad (\forall j, \forall a_j).
\end{align}
For a potential game, since this condition is rewritten as
\begin{align}
 \Phi \left( a_j^*, a_{-j}^* \right) &\geq \Phi \left( a_j, a_{-j}^* \right) \quad (\forall j, \forall a_j),
\end{align}
Nash equilibria are realized as local maxima of a potential function.
A sequence of actions improving the payoffs can be regarded as that increasing the value of a potential function.

For subsets $\forall \mathcal{N}^\prime \subseteq \mathcal{N}$, we write $a_{\mathcal{N}^\prime} := \left( a_j \right)_{j \in \mathcal{N}^\prime}$ and $A_{\mathcal{N}^\prime} := \prod_{j \in \mathcal{N}^\prime} A_j$.
Ui proved that a potential function can be described as the sum of interaction potentials.
\begin{proposition}[\cite{Ui2000}]
\label{prop:interaction}
$G$ is a potential game if and only if there exist interaction potentials $\left\{ \phi_{\mathcal{N}^\prime} | \phi_{\mathcal{N}^\prime}: A_{\mathcal{N}^\prime} \rightarrow \mathbb{R}, \mathcal{N}^\prime \subseteq \mathcal{N}  \right\}$ such that
\begin{align}
 s_j \left( \bm{a} \right) &= \sum_{\mathcal{N}^\prime \subseteq \mathcal{N}}^{j \in \mathcal{N}^\prime} \phi_{\mathcal{N}^\prime} \left( a_{\mathcal{N}^\prime} \right)
\end{align}
for all $j \in \mathcal{N}$ and all $\bm{a} \in \mathcal{A}$.
A potential function is given by
\begin{align}
 \Phi \left( \bm{a} \right) &= \sum_{\mathcal{N}^\prime \subseteq \mathcal{N}} \phi_{\mathcal{N}^\prime} \left( a_{\mathcal{N}^\prime} \right).
\end{align}
\end{proposition}

\section{Realization of ZD strategies by one-dimensional transition probability}
\label{sec:physicalZDS}
\subsection{Main result}
First, we prove that, if ZD strategies exist, they can be implemented by one-dimensional transition probability.
Here, we call transition probability of $a_j$ \emph{one-dimensional} when transition of $a_j$ occurs only between nearest neighbor sites in $a_j$-coordinate.
We consider the situation that a ZD strategy of player $j$ controlling $B$ exists.
We can relabel the name of actions of player $j$ from $\left\{ b_j^{(1)}, \cdots, b_j^{(L_j)} \right\}$ to $\left\{ 1, \cdots, L_j \right\}$ so as to satisfy $\overline{a}_j=1$ and $\underline{a}_j=L_j$ in Eq. (\ref{eq:condition_exsitence}).
(The rest of actions is renamed arbitrarily.)
Below we redefine $A_j=\left\{ 1, \cdots, L_j \right\}$.

\begin{theorem}
\label{thm:physicalZDS}
If a ZD strategy of player $j$ controlling $B$ exists, it can be realized by one-dimensional transition probability.
\end{theorem}

\begin{proof}
According to Proposition \ref{prop:existence}, a necessary and sufficient condition for the existence of a ZD strategy is given by Eq. (\ref{eq:condition_exsitence}), that is,
\begin{align}
 B \left( 1, a_{-j} \right) &\geq 0 \quad \left( \forall a_{-j} \right) \nonumber \\
 B \left( L_j, a_{-j} \right) &\leq 0 \quad \left( \forall a_{-j} \right)
 \label{eq:condition_exsitence_phys}
\end{align}
and $B$ is not identically zero.
Below we write the total number of action profiles as $M:= \prod_{k=1}^N L_k$.
For simplicity, we introduce a vector notation of the quantity $D:\mathcal{A}\rightarrow \mathbb{R}$ by $\bm{D}:= \left( D(\bm{a}) \right)_{\bm{a}\in \mathcal{A}} \in \mathbb{R}^M$.
It should be noted that a trivial identity
\begin{align}
 D(\bm{a}^\prime) &= \sum_{a_j \in A_j} \sum_{d=\pm 1} D(\bm{a}^\prime) \mathbb{I}(d D(\bm{a}^\prime)> 0) \mathbb{I}(a_j^\prime=a_j) 
\end{align}
holds for all $\bm{a}^\prime \in \mathcal{A}$, where $\mathbb{I}(\cdots)$ is an indicator function which returns $1$ if $\cdots$ holds, and $0$ otherwise.
By using the above vector notation, this identity is rewritten as
\begin{align}
 \bm{D} &= \sum_{a_j} \sum_{d=\pm 1} \left[ \bm{D} \right]_{a_j, \mathrm{sgn}(d)},
\end{align}
where we have defined
\begin{align}
 \left[ \bm{D} \right]_{a_j, \mathrm{sgn}(d)} &:= \left( D(\bm{a}^\prime) \mathbb{I}(d D(\bm{a}^\prime)> 0) \mathbb{I}(a_j^\prime=a_j) \right)_{\bm{a}^\prime \in \mathcal{A}} \nonumber \\
 & \quad \left( a_j\in A_j, d\in \{ 1,-1 \} \right).
 \label{eq:def_decomposition}
\end{align} 
For the quantity $\bm{B}\in \mathbb{R}^M$, our assumption (\ref{eq:condition_exsitence_phys}) implies
\begin{align}
 \bm{B} &= \sum_{a_j\neq 1, L_j} \sum_{d=\pm 1} \left[ \bm{B} \right]_{a_j, \mathrm{sgn}(d)} + \left[ \bm{B} \right]_{L_j, -} + \left[ \bm{B} \right]_{1, +}.
 \label{eq:decompose_B}
\end{align}
We also collectively write the Press-Dyson vectors of player $j$ by $\hat{\bm{T}}_j\left( a_j \right) := \left( \hat{T}_j\left( a_j | \bm{a}^{\prime} \right) \right)_{\bm{a}^\prime \in \mathcal{A}}$ for all $a_j \in A_j$.

We remark that the quantity
\begin{align}
 W &:= \max_{\bm{a}\in \mathcal{A}} \left| B(\bm{a}) \right|
 \label{eq:W}
\end{align}
satisfies $W\neq 0$ because of our assumption that $B$ is not identically zero.
We now consider the following strategy of player $j$:
\begin{align}
 \hat{\bm{T}}_j\left( 1 \right) &= \frac{1}{W} \left( - \left[ \bm{B} \right]_{1, +} - \left[ \bm{B} \right]_{2, -} \right) \nonumber \\
 \hat{\bm{T}}_j\left( a_j \right) &= \frac{1}{W} \left( - \left[ \bm{B} \right]_{a_j, +} + \left[ \bm{B} \right]_{a_j, -} - \left[ \bm{B} \right]_{a_j+1, -} + \left[ \bm{B} \right]_{a_j-1, +} \right) \nonumber \\
 & \quad \left( 2\leq a_j \leq L_j-1 \right) \nonumber \\
 \hat{\bm{T}}_j\left( L_j \right) &= \frac{1}{W} \left( \left[ \bm{B} \right]_{L_j, -} + \left[ \bm{B} \right]_{L_j-1, +} \right).
 \label{eq:strategy_physZD}
\end{align}
It should be noted that this strategy is different from a strategy originally introduced in Ref. \cite{Ued2022b} for the proof of Proposition \ref{prop:existence}, which is not one-dimensional.
We check that Eq. (\ref{eq:strategy_physZD}) is indeed Press-Dyson vectors of some strategy.
By the definitions (\ref{eq:def_decomposition}) and (\ref{eq:W}), it satisfies Eq. (\ref{eq:PD_range}).
Furthermore, we obtain
\begin{align}
 \sum_{a_j} \hat{\bm{T}}_j\left( a_j \right) &= \frac{1}{W} \left( - \sum_{a_j=1}^{L_j-1} \left[ \bm{B} \right]_{a_j, +} +  \sum_{a_j=2}^{L_j} \left[ \bm{B} \right]_{a_j, -} \right. \nonumber \\
 & \left.  -  \sum_{a_j=1}^{L_j-1} \left[ \bm{B} \right]_{a_j+1, -} +  \sum_{a_j=2}^{L_j} \left[ \bm{B} \right]_{a_j-1, +} \right) \nonumber \\
 &= \bm{0},
\end{align}
which implies Eq. (\ref{eq:PD_normalized}).

Next, we show that Eq. (\ref{eq:strategy_physZD}) is a ZD strategy.
By calculating
\begin{align}
 \sum_{a_j} a_j \hat{\bm{T}}_j\left( a_j \right) &= \frac{1}{W} \left( - \sum_{a_j=1}^{L_j-1} a_j \left[ \bm{B} \right]_{a_j, +} +  \sum_{a_j=2}^{L_j} a_j \left[ \bm{B} \right]_{a_j, -} \right. \nonumber \\
 & \left.  -  \sum_{a_j=1}^{L_j-1} a_j \left[ \bm{B} \right]_{a_j+1, -} +  \sum_{a_j=2}^{L_j} a_j \left[ \bm{B} \right]_{a_j-1, +} \right) \nonumber \\
 &= \frac{1}{W} \left( \sum_{a_j=1}^{L_j-1} \left[ \bm{B} \right]_{a_j, +} +  \sum_{a_j=2}^{L_j} \left[ \bm{B} \right]_{a_j, -} \right) \nonumber \\
 &= \frac{1}{W} \bm{B},
\end{align}
we can check that the strategy is indeed a ZD strategy (Definition \ref{def:ZDS}).

Finally, we check that Eq. (\ref{eq:strategy_physZD}) is realized by one-dimensional transition probability.
It is trivial from the definition (\ref{eq:def_decomposition}) and the definition of the Press-Dyson vectors.
Explicitly, each component of the transition probability of the ZD strategy (\ref{eq:strategy_physZD}) is
\begin{align}
 T_j\left( 1 | \bm{a}^\prime \right) &= \delta_{1, a^\prime_j} + \frac{1}{W} \left( - \left| B\left( \bm{a}^\prime \right) \right| \delta_{1, a^\prime_j} + \left| B\left( \bm{a}^\prime \right) \right| \mathbb{I}(B(\bm{a}^\prime)< 0) \delta_{2, a^\prime_j} \right) \nonumber \\
 T_j\left( a_j | \bm{a}^\prime \right) &= \delta_{a_j, a^\prime_j} + \frac{1}{W} \left( - \left| B\left( \bm{a}^\prime \right) \right| \delta_{a_j, a^\prime_j} + \left| B\left( \bm{a}^\prime \right) \right| \mathbb{I}(B(\bm{a}^\prime)< 0) \delta_{a_j+1, a^\prime_j} \right. \nonumber \\
 & \qquad \left. + \left| B\left( \bm{a}^\prime \right) \right| \mathbb{I}(B(\bm{a}^\prime)> 0) \delta_{a_j-1, a^\prime_j} \right) \nonumber \\
 & \quad \left( 2\leq a_j \leq L_j-1 \right) \nonumber \\
 T_j\left( L_j | \bm{a}^\prime \right) &= \delta_{L_j, a^\prime_j} + \frac{1}{W} \left( - \left| B\left( \bm{a}^\prime \right) \right| \delta_{L_j, a^\prime_j} + \left| B\left( \bm{a}^\prime \right) \right| \mathbb{I}(B(\bm{a}^\prime)> 0) \delta_{L_j-1, a^\prime_j} \right),
 \label{eq:strategy_physZD_comp}
\end{align}
which is one-dimensional.
\end{proof}

We remark that $\left\{ C\hat{\bm{T}}_j (a_j) \right\}_{a_j \in A_j}$ with $0<C\leq 1$, where $\left\{ \hat{\bm{T}}_j(a_j) \right\}_{a_j \in A_j}$ is given by Eq. (\ref{eq:strategy_physZD}), is also a ZD strategy, due to Definition \ref{def:ZDS}.
In addition, the ZD strategy (\ref{eq:strategy_physZD_comp}) increases $a_j$ when $B(\bm{a}^\prime)>0$ in the previous round and decreases $a_j$ when $B(\bm{a}^\prime)<0$ in the previous round under the boundary condition (\ref{eq:condition_exsitence_phys}).
This fact, which is similar to feedback control, provides an interpretation why a ZD strategy can unilaterally enforce Eq. (\ref{eq:linear}).
We also remark that each element of the transition probability (\ref{eq:strategy_physZD_comp}) depends only on the quantity at the previous state $\bm{a}^\prime$, like the trap model \cite{BouGeo1990}.

\subsection{Continuum limit}
Although informal, here we consider the continuum limit of this one-dimensional stochastic process.
When all players use time-independent memory-one strategies, the time evolution is described by the Markov chain
\begin{align}
 P_{t+1} \left( \bm{a} \right) &= \sum_{\bm{a}^\prime} T\left( \bm{a} | \bm{a}^\prime \right) P_{t} \left( \bm{a}^\prime \right)
 \label{eq:MC}
\end{align}
with the total transition probability
\begin{align}
  T\left( \bm{a} | \bm{a}^\prime \right) &:= \prod_{j=1}^N T_j\left( a_j | \bm{a}^\prime \right).
  \label{eq:transition}
\end{align}

We first consider the continuous time limit and then the continuous space limit.
For this purpose, we assume that transition in $T_j$ occurs with probability of order $\Delta t$ for all $j$:
\begin{align}
  T_j\left( a_j | \bm{a}^\prime \right) &= \delta_{a_j, a^\prime_j} + \Delta t R_j \left( a_j | \bm{a}^\prime \right),
  \label{eq:rate}
\end{align}
where $R_j$ is the transition rate.
It should be noted that this assumption provides the meaning of the Press-Dyson vectors (\ref{eq:PD}):
\begin{align}
  \hat{T}_j\left( a_j | \bm{a}^\prime \right) &= \Delta t R_j \left( a_j | \bm{a}^\prime \right).
\end{align}
According to Definition \ref{def:ZDS}, as far as ZD strategies are concerned, multiplying a constant to the Press-Dyson vectors does not change properties of ZD strategies.
Therefore, we can identify $\hat{T}_j$ and $R_j$.

By taking $\Delta t \rightarrow 0$ in Eq. (\ref{eq:MC}), we obtain
\begin{align}
  \frac{\partial P}{\partial t}\left( \bm{a}; t \right) &= \sum_{j=1}^N \sum_{a^\prime_j} R_j\left( a_j | a^\prime_j, a_{-j} \right) P\left( a^\prime_j, a_{-j}; t \right),
  \label{eq:CTMP}
\end{align}
where $P\left( \bm{a}; t \right)$ is the probability distribution of $\bm{a}$ at real time $t\in \mathbb{R}$.
We remark that the transition rate satisfies
\begin{align}
  \sum_{a_j} R_j \left( a_j | \bm{a}^\prime \right) &= 0 \quad (\forall j, \forall \bm{a}^\prime ),
\end{align}
due to the normalization condition of Eq. (\ref{eq:rate}).

Next, we assume that player $j$ uses a ZD strategy, and consider the continuous space limit of Eq. (\ref{eq:CTMP}) with respect to $a_j$.
When $R_j$ is given by the ZD strategy (\ref{eq:strategy_physZD}) with $1/W$ replaced by $\rho>0$, the $j$-th term in the right-hand side of Eq. (\ref{eq:CTMP}) is explicitly written as
\begin{align}
 & \sum_{a^\prime_j} R_j\left( a_j | a^\prime_j, a_{-j} \right) P\left( a^\prime_j, a_{-j}; t \right) \nonumber \\
 &= \rho \sum_{a^\prime_j} \left[ - \left| B\left( a^\prime_j, a_{-j} \right) \right| \delta_{a_j, a^\prime_j} + \left| B\left( a^\prime_j, a_{-j} \right) \right| \mathbb{I}\left( B\left( a^\prime_j, a_{-j} \right) < 0 \right)  \delta_{a_j+1, a^\prime_j} \right. \nonumber \\
 & \quad \left. + \left| B\left( a^\prime_j, a_{-j} \right) \right| \mathbb{I}\left( B\left( a^\prime_j, a_{-j} \right) > 0 \right)  \delta_{a_j-1, a^\prime_j} \right] P\left( a^\prime_j, a_{-j}; t \right) \nonumber \\
 &= \rho \left[ - \left| B\left( a_j, a_{-j} \right) \right| P\left( a_j, a_{-j}; t \right) \right. \nonumber \\
 & \quad + \left| B\left( a_j+1, a_{-j} \right) \right| \mathbb{I}\left( B\left( a_j+1, a_{-j} \right) < 0 \right)  P\left( a_j+1, a_{-j}; t \right) \nonumber \\
 & \quad \left. + \left| B\left( a_j-1, a_{-j} \right) \right| \mathbb{I}\left( B\left( a_j-1, a_{-j} \right) > 0 \right)  P\left( a_j-1, a_{-j}; t \right) \right]
 \label{eq:CSTMP_R1}
\end{align}
for $2\leq a_j \leq L_j-1$.
Here, we regard the length between two nearest neighbor sites in $a_j$-coordinate as $\Delta a$, and define $z_j := a_j \Delta a$.
Furthermore, we redefine the probability distribution of $(z_j, a_{-j})$ by $\tilde{P}(z_j, a_{-j}; t):=P(a_j, a_{-j}; t)/\Delta a$, and the quantity $B$ by $\tilde{B}(z_j, a_{-j}):=B(a_j, a_{-j})$.
We assume that the function $\tilde{B}$ is smooth with respect to $z_j$ in the limit $\Delta a \rightarrow 0$ with $\tilde{L}_j:=L_j\Delta a$ and $\tilde{\rho}:= \rho \Delta a$ fixed.
Then, Eq. (\ref{eq:CSTMP_R1}) is rewritten as
\begin{align}
 & \sum_{a^\prime_j} R_j\left( a_j | a^\prime_j, a_{-j} \right) \tilde{P}\left( a^\prime_j\Delta a, a_{-j}; t \right) \nonumber \\
 &= \rho \left[ - \left| \tilde{B}\left( z_j, a_{-j} \right) \right| \tilde{P}\left( z_j, a_{-j}; t \right) \right. \nonumber \\
 & \quad + \left| \tilde{B}\left( z_j+\Delta a, a_{-j} \right) \right| \mathbb{I}\left( \tilde{B}\left( z_j+\Delta a, a_{-j} \right) < 0 \right)  \tilde{P}\left( z_j+\Delta a, a_{-j}; t \right) \nonumber \\
 & \quad \left. + \left| \tilde{B}\left( z_j-\Delta a, a_{-j} \right) \right| \mathbb{I}\left( \tilde{B}\left( z_j-\Delta a, a_{-j} \right) > 0 \right) \tilde{P}\left( z_j-\Delta a, a_{-j}; t \right) \right]\nonumber \\
 &= \rho \left[ - \Delta a \frac{\partial}{\partial z_j} \left( \tilde{B}\left( z_j, a_{-j} \right) \tilde{P}\left( z_j, a_{-j}; t \right) \right) \right. \nonumber \\
 & \quad \left. + \frac{\Delta a^2}{2} \frac{\partial^2}{\partial z_j^2} \left( \left| \tilde{B}\left( z_j, a_{-j} \right) \right| \tilde{P}\left( z_j, a_{-j}; t \right) \right) + \mathcal{O}(\Delta a^3) \right] \\
 &\rightarrow -\tilde{\rho} \frac{\partial}{\partial z_j} \left( \tilde{B}\left( z_j, a_{-j} \right) \tilde{P}\left( z_j, a_{-j}; t \right) \right) \quad (\Delta a \rightarrow 0).
\end{align}
Therefore, in the limit $\Delta a \rightarrow 0$, we find that the dynamics of $z_j$ is described by the time evolution equation:
\begin{align}
 \dot{z}_j &= \tilde{\rho} \tilde{B}\left( z_j, a_{-j} \right).
\end{align}
Note that, when we calculate $\lim_{\tau \rightarrow \infty}1/\tau \int_0^\tau dt$ of both sides, this equation becomes
\begin{align}
 0 &= \tilde{\rho} \lim_{\tau \rightarrow \infty} \frac{1}{\tau} \int_0^\tau dt \tilde{B}\left( z_j(t), a_{-j}(t) \right)
\end{align}
because $0\leq z_j \leq \tilde{L}_j$.
This is why our ZD strategy unilaterally enforces $\left\langle B \right\rangle^{*}=0$ irrespective of strategies of other players.

\section{Existence of ZD strategies in two-player potential games}
\label{sec:ZDS_potential}
\subsection{Main result}
Next, we prove that ZD strategies controlling payoffs exist in two-player non-trivial potential games.
According to Proposition \ref{prop:interaction}, for two-player potential games, the payoffs and a potential is explicitly described as
\begin{align}
 s_1 \left( \bm{a} \right) &= \phi_{\{ 1 \}}\left( a_1 \right) + \phi_{\{ 1, 2 \}}\left( a_1, a_2 \right) \nonumber \\
 s_2 \left( \bm{a} \right) &= \phi_{\{ 2 \}}\left( a_2 \right) + \phi_{\{ 1, 2 \}}\left( a_1, a_2 \right)
 \label{eq:payoff_interaction}
\end{align}
and
\begin{align}
 \Phi \left( \bm{a} \right) &= \phi_{\emptyset} +  \phi_{\{ 1 \}}\left( a_1 \right) + \phi_{\{ 2 \}}\left( a_2 \right) + \phi_{\{ 1, 2 \}}\left( a_1, a_2 \right).
 \label{eq:potential_interaction}
\end{align}

\begin{theorem}
\label{thm:ZDS_potential}
For every two-player potential game with $\phi_{\{ 1 \}} \neq \mathrm{const.}$ or $\phi_{\{ 2 \}} \neq \mathrm{const.}$, there exists a ZD strategy controlling
\begin{align}
 s_1\left( \bm{a} \right) - s_2\left( \bm{a} \right) - \frac{\overline{\phi}_{\{ 1 \}} + \underline{\phi}_{\{ 1 \}}}{2} + \frac{\overline{\phi}_{\{ 2 \}} + \underline{\phi}_{\{ 2 \}}}{2},
\end{align}
where we have defined
\begin{align}
 \overline{\phi}_{\{ j \}} &:= \max_{a_j} \phi_{\{ j \}}\left( a_j \right) \quad (\forall j) \nonumber \\
 \underline{\phi}_{\{ j \}} &:= \min_{a_j} \phi_{\{ j \}}\left( a_j \right) \quad (\forall j).
 \label{eq:phibar}
\end{align}
\end{theorem}

\begin{proof}
We first define
\begin{align}
 \overline{a}_{j} &:= \arg\max_{a_j} \phi_{\{ j \}}\left( a_j \right) \quad (\forall j) \nonumber \\
 \underline{a}_{j} &:= \arg\min_{a_j} \phi_{\{ j \}}\left( a_j \right) \quad (\forall j),
 \label{eq:abar_potential}
\end{align}
where ties may be broken arbitrarily.
We consider the following two cases separately.
\begin{enumerate}[(i)]
 \item $\overline{\phi}_{\{ 1 \}} - \underline{\phi}_{\{ 1 \}} \geq \overline{\phi}_{\{ 2 \}} - \underline{\phi}_{\{ 2 \}}$\\
We consider the quantity
\begin{align}
 B \left( \bm{a} \right) &:= s_1\left( \bm{a} \right) - s_2\left( \bm{a} \right) - \frac{\overline{\phi}_{\{ 1 \}} + \underline{\phi}_{\{ 1 \}}}{2} + \frac{\overline{\phi}_{\{ 2 \}} + \underline{\phi}_{\{ 2 \}}}{2}.
\end{align}
By using the property of potential games (\ref{eq:payoff_interaction}), this quantity is rewritten as
\begin{align}
 B \left( \bm{a} \right) &= \phi_{\{ 1 \}}\left( a_1 \right) - \phi_{\{ 2 \}}\left( a_2 \right) - \frac{\overline{\phi}_{\{ 1 \}} + \underline{\phi}_{\{ 1 \}}}{2} + \frac{\overline{\phi}_{\{ 2 \}} + \underline{\phi}_{\{ 2 \}}}{2}.
\end{align}
Then we find that
\begin{align}
 B \left( \overline{a}_{1}, a_2 \right) &= \frac{\overline{\phi}_{\{ 1 \}} - \underline{\phi}_{\{ 1 \}}}{2} - \phi_{\{ 2 \}}\left( a_2 \right) + \frac{\overline{\phi}_{\{ 2 \}} + \underline{\phi}_{\{ 2 \}}}{2} \\
 B \left( \underline{a}_{1}, a_2 \right) &= - \frac{\overline{\phi}_{\{ 1 \}} - \underline{\phi}_{\{ 1 \}}}{2} - \phi_{\{ 2 \}}\left( a_2 \right) + \frac{\overline{\phi}_{\{ 2 \}} + \underline{\phi}_{\{ 2 \}}}{2}.
\end{align}
By the definition of the maximum and the minimum, we obtain
\begin{align}
 B \left( \overline{a}_{1}, a_2 \right) &\geq B \left( \overline{a}_{1}, \overline{a}_2 \right) \nonumber \\
 &= \frac{\overline{\phi}_{\{ 1 \}} - \underline{\phi}_{\{ 1 \}}}{2} - \frac{\overline{\phi}_{\{ 2 \}} - \underline{\phi}_{\{ 2 \}}}{2} \nonumber \\
 &\geq 0 \\
 B \left( \underline{a}_{1}, a_2 \right) &\leq B \left( \underline{a}_{1}, \underline{a}_2 \right) \nonumber \\
 &= - \frac{\overline{\phi}_{\{ 1 \}} - \underline{\phi}_{\{ 1 \}}}{2} + \frac{\overline{\phi}_{\{ 2 \}} - \underline{\phi}_{\{ 2 \}}}{2} \nonumber \\
 &\leq 0
\end{align}
for all $a_2$.
Furthermore, $B \left( \bm{a} \right)$ is identically zero only if $\phi_{\{ 1 \}}$ and $\phi_{\{ 2 \}}$ are constant functions.
Such situations are excluded by the assumption of theorem.
Therefore, the two actions $\overline{a}_{1}$ and $\underline{a}_{1}$ can be regarded as those in Proposition \ref{prop:existence}, and a ZD strategy of player $1$ exists.

 \item $\overline{\phi}_{\{ 1 \}} - \underline{\phi}_{\{ 1 \}} \leq \overline{\phi}_{\{ 2 \}} - \underline{\phi}_{\{ 2 \}}$\\
We consider the quantity
\begin{align}
 B \left( \bm{a} \right) &:= s_2\left( \bm{a} \right) - s_1\left( \bm{a} \right) - \frac{\overline{\phi}_{\{ 2 \}} + \underline{\phi}_{\{ 2 \}}}{2} + \frac{\overline{\phi}_{\{ 1 \}} + \underline{\phi}_{\{ 1 \}}}{2}.
\end{align}
By using the property of potential games (\ref{eq:payoff_interaction}), this quantity is rewritten as
\begin{align}
 B \left( \bm{a} \right) &= \phi_{\{ 2 \}}\left( a_2 \right) - \phi_{\{ 1 \}}\left( a_1 \right) - \frac{\overline{\phi}_{\{ 2 \}} + \underline{\phi}_{\{ 2 \}}}{2} + \frac{\overline{\phi}_{\{ 1 \}} + \underline{\phi}_{\{ 1 \}}}{2}.
\end{align}
Then we find that
\begin{align}
 B \left( a_1, \overline{a}_{2} \right) &= \frac{\overline{\phi}_{\{ 2 \}} - \underline{\phi}_{\{ 2 \}}}{2} - \phi_{\{ 1 \}}\left( a_1 \right) + \frac{\overline{\phi}_{\{ 1 \}} + \underline{\phi}_{\{ 1 \}}}{2} \\
 B \left( a_1, \underline{a}_{2} \right) &= - \frac{\overline{\phi}_{\{ 2 \}} - \underline{\phi}_{\{ 2 \}}}{2} - \phi_{\{ 1 \}}\left( a_1 \right) + \frac{\overline{\phi}_{\{ 1 \}} + \underline{\phi}_{\{ 1 \}}}{2}.
\end{align}
By the definition of the maximum and the minimum, we obtain
\begin{align}
 B \left( a_1, \overline{a}_{2} \right) &\geq B \left( \overline{a}_{1}, \overline{a}_2 \right) \nonumber \\
 &= \frac{\overline{\phi}_{\{ 2 \}} - \underline{\phi}_{\{ 2 \}}}{2} - \frac{\overline{\phi}_{\{ 1 \}} - \underline{\phi}_{\{ 1 \}}}{2} \nonumber \\
 &\geq 0 \\
 B \left( a_1, \underline{a}_{2} \right) &\leq B \left( \underline{a}_{1}, \underline{a}_2 \right) \nonumber \\
 &= - \frac{\overline{\phi}_{\{ 2 \}} - \underline{\phi}_{\{ 2 \}}}{2} + \frac{\overline{\phi}_{\{ 1 \}} - \underline{\phi}_{\{ 1 \}}}{2} \nonumber \\
 &\leq 0
\end{align}
for all $a_1$.
Furthermore, $B \left( \bm{a} \right)$ is identically zero only if $\phi_{\{ 1 \}}$ and $\phi_{\{ 2 \}}$ are constant functions.
Such situations are excluded by the assumption of theorem.
Therefore, the two actions $\overline{a}_{2}$ and $\underline{a}_{2}$ can be regarded as those in Proposition \ref{prop:existence}, and a ZD strategy of player $2$ exists.
\end{enumerate}
Thus, we conclude that at least one ZD strategy exists under the assumption.
We note that a linear relation enforced by the ZD strategy is 
\begin{align}
 \left\langle s_{1} \right\rangle^{*} - \frac{\overline{\phi}_{\{ 1 \}} + \underline{\phi}_{\{ 1 \}}}{2} &= \left\langle s_{2} \right\rangle^{*} - \frac{\overline{\phi}_{\{ 2 \}} + \underline{\phi}_{\{ 2 \}}}{2},
 \label{eq:linear_potential}
\end{align}
or equivalently,
\begin{align}
 \left\langle \phi_{\{ 1 \}} \right\rangle^{*} - \frac{\overline{\phi}_{\{ 1 \}} + \underline{\phi}_{\{ 1 \}}}{2} &= \left\langle \phi_{\{ 2 \}} \right\rangle^{*} - \frac{\overline{\phi}_{\{ 2 \}} + \underline{\phi}_{\{ 2 \}}}{2}.
 \label{eq:linear_potential_mod}
\end{align}
\end{proof}

We remark that for the case $\phi_{\{ 1 \}} = \mathrm{const.}$ and $\phi_{\{ 2 \}} = \mathrm{const.}$, the payoffs and the potential are
\begin{align}
 s_1 \left( \bm{a} \right) &= \phi_{1,0} + \phi_{\{ 1, 2 \}}\left( a_1, a_2 \right) \nonumber \\
 s_2 \left( \bm{a} \right) &= \phi_{2,0} + \phi_{\{ 1, 2 \}}\left( a_1, a_2 \right)
\end{align}
and
\begin{align}
 \Phi \left( \bm{a} \right) &= \phi_{0} + \phi_{\{ 1, 2 \}}\left( a_1, a_2 \right).
\end{align}
Such situation is called as a common interest game.
For common interest games, vectors $\{ \bm{s}_1, \bm{s}_2, \bm{1} \}$ are not linearly independent, where $\bm{1}$ is the vector of all ones.
Additional analysis is needed for investigating the existence of ZD strategies controlling payoffs for common interest games.
For example, a two-player two-action common interest game $\bm{s}_1=\bm{s}_2=\bm{\Phi}=(1, -1, -1, 1)^\mathsf{T}$ does not contain any ZD strategies controlling payoffs.

We also note that the linear relation (\ref{eq:linear_potential}) may be regarded as an extension of fairness condition $\left\langle s_{1} \right\rangle^{*} = \left\langle s_{2} \right\rangle^{*}$ in two-player symmetric potential games \cite{Ued2022}.
Therefore, such ZD strategies may be regarded as an extension of unbeatable strategies in two-player symmetric games to asymmetric cases.

\subsection{Asymmetric prisoner's dilemma}
As an example, we consider an asymmetric prisoner's dilemma game.
The game is described as $N=2$, $A_1=A_2=\{C, D\}$, and
\begin{align}
 \bm{s}_1 &:= \left( s_1(C,C), s_1(C,D), s_1(D,C), s_1(D,D) \right)^\mathsf{T} = \left( R_1, S_1, T_1, P_1 \right)^\mathsf{T} \nonumber \\
 \bm{s}_2 &:= \left( s_2(C,C), s_2(C,D), s_2(D,C), s_2(D,D) \right)^\mathsf{T} = \left( R_2, T_2, S_2, P_2 \right)^\mathsf{T}
\end{align}
with $T_j>R_j>P_j>S_j$ $(j=1,2)$.
(The actions $C$ and $D$ imply cooperation and defection, respectively.)
If this game is a potential game, according to Definition \ref{def:potential}, a potential function $\Phi$ satisfies
\begin{align}
 R_1 - T_1 &= \Phi(C,C) - \Phi(D,C) \nonumber \\
 S_1 - P_1 &= \Phi(C,D) - \Phi(D,D) \nonumber \\
 R_2 - T_2 &= \Phi(C,C) - \Phi(C,D) \nonumber \\
 S_2 - P_2 &= \Phi(D,C) - \Phi(D,D).
\end{align}
Therefore, for such $\Phi$ to exist, the payoffs must satisfy
\begin{align}
 R_1 - T_1 - S_1 + P_1 &= R_2 - T_2 - S_2 + P_2.
\end{align}
Under this condition, the potential is determined as
\begin{align}
 \Phi(C,C) &= \Phi_0 \nonumber \\
 \Phi(C,D) &= \Phi_0 + T_2 - R_2 \nonumber \\
 \Phi(D,C) &= \Phi_0 + T_1 - R_1 \nonumber \\
 \Phi(D,D) &= \Phi_0 + T_2 - R_2 + P_1 - S_1,
\end{align}
where $\Phi_0$ is an arbitrary constant.

According to Eqs. (\ref{eq:payoff_interaction}) and (\ref{eq:potential_interaction}),
\begin{align}
 \Phi \left( \bm{a} \right) - s_1\left( \bm{a} \right) &= \phi_{\emptyset} + \phi_{\{ 2 \}}\left( a_2 \right) \nonumber \\
 \Phi \left( \bm{a} \right) - s_2\left( \bm{a} \right) &= \phi_{\emptyset} + \phi_{\{ 1 \}}\left( a_1 \right).
\end{align}
It should be noted that such property that the right-hand sides depend only on the opponent's action is called \emph{coordination-dummy separability} \cite{LCS2016}.
Then, we obtain
\begin{align}
 \phi_{\{ 1 \}}\left( C \right) &= \Phi_0 - R_2 - \phi_{\emptyset} \nonumber \\
 \phi_{\{ 1 \}}\left( D \right) &= \Phi_0 + T_1 - R_1 - S_2 - \phi_{\emptyset} \nonumber \\
 \phi_{\{ 2 \}}\left( C \right) &= \Phi_0 - R_1 - \phi_{\emptyset} \nonumber \\
 \phi_{\{ 2 \}}\left( D \right) &= \Phi_0 + T_2 - R_2 - S_1 - \phi_{\emptyset}.
\end{align}
By using the definition (\ref{eq:phibar}) and (\ref{eq:abar_potential}), we find $\overline{a}_{1}=\overline{a}_{2}=D$, $\underline{a}_{1}=\underline{a}_{2}=C$, and
\begin{align}
 \overline{\phi}_{\{ 1 \}} &= \Phi_0 + T_1 - R_1 - S_2 - \phi_{\emptyset} \nonumber \\
 \underline{\phi}_{\{ 1 \}} &= \Phi_0 - R_2 - \phi_{\emptyset} \nonumber \\
 \overline{\phi}_{\{ 2 \}} &= \Phi_0 + T_2 - R_2 - S_1 - \phi_{\emptyset} \nonumber \\
 \underline{\phi}_{\{ 2 \}} &= \Phi_0 - R_1 - \phi_{\emptyset}.
\end{align}
Therefore, according to Theorem \ref{thm:ZDS_potential}, a ZD strategy unilaterally enforcing
\begin{align}
 \left\langle s_{1} \right\rangle^{*} - \frac{T_1+S_1}{2} &= \left\langle s_{2} \right\rangle^{*} - \frac{T_2+S_2}{2}
\end{align}
exists.

We remark that, as shown in the proof of Theorem \ref{thm:ZDS_potential}, which player can use the ZD strategy depends on the values $\overline{\phi}_{\{ 1 \}} - \underline{\phi}_{\{ 1 \}}$ and $\overline{\phi}_{\{ 2 \}} - \underline{\phi}_{\{ 2 \}}$.
For the ZD strategy of player $1$ to exist, the inequality $\overline{\phi}_{\{ 1 \}} - \underline{\phi}_{\{ 1 \}} \geq \overline{\phi}_{\{ 2 \}} - \underline{\phi}_{\{ 2 \}}$ must hold, which is equivalent to
\begin{align}
 R_2 - P_2 &\geq R_1 - P_1.
 \label{eq:ARPD_ZDS1}
\end{align}
Similarly, for the ZD strategy of player $2$ to exist, the inequality
\begin{align}
 R_1 - P_1 &\geq R_2 - P_2
 \label{eq:ARPD_ZDS2}
\end{align}
must hold.
The concrete transition probability of the ZD strategy is constructed by using Theorem \ref{thm:physicalZDS}.
Explicitly, when the inequality (\ref{eq:ARPD_ZDS1}) holds, the ZD strategy of player $1$ is given by
\begin{align}
 \bm{T}_1(C) &:= \left(
\begin{array}{c}
T_1(C|C,C)  \\
T_1(C|C,D)  \\
T_1(C|D,C)  \\
T_1(C|D,D)
\end{array}
\right) \nonumber \\
&= \left(
\begin{array}{c}
1  \\
1  \\
0  \\
0
\end{array}
\right)
+ \frac{1}{W} \left(
\begin{array}{c}
\frac{R_1-P_1-R_2+P_2}{2}  \\
\frac{S_1-T_1-T_2+S_2}{2}  \\
\frac{T_1-S_1-S_2+T_2}{2}  \\
\frac{P_1-R_1-P_2+R_2}{2}
\end{array}
\right),
\end{align}
where
\begin{align}
 W &= \max \left\{ \left| \frac{R_1-P_1-R_2+P_2}{2} \right|, \left| \frac{S_1-T_1-T_2+S_2}{2} \right| \right\}.
\end{align}
It should be noted that, for the symmetric case $\left( R_1, S_1, T_1, P_1 \right)=\left( R_2, S_2, T_2, P_2 \right)$, this ZD strategy is reduced to the Tit-for-Tat strategy \cite{PreDys2012}.
The ZD strategy of player $2$ when the inequality (\ref{eq:ARPD_ZDS2}) holds is obtained similarly.

\section{Discussion}
\label{sec:discussion}
In this section, we discuss a relation between our Theorem \ref{thm:physicalZDS}  and previous studies.
In Ref. \cite{McAHau2016}, McAvoy and Hauert proved the following proposition, which holds even if the action spaces are arbitrary (not limited to finite sets).
\begin{proposition}[\cite{McAHau2016}]
\label{prop:two-point}
Suppose that there exist two actions $\underline{a}_j, \overline{a}_j \in A_j$ and $W>0$ such that
\begin{align}
 -W &\leq B\left( \underline{a}_j, a_{-j} \right) \leq 0 \quad (\forall a_{-j} \in A_{-j}) \nonumber \\
 0 &\leq B\left( \overline{a}_j, a_{-j} \right) \leq W \quad (\forall a_{-j} \in A_{-j}).
 \label{eq:existence_condition_MH}
\end{align}
Then, when we restrict the action set of player $j$ from $A_j$ to $A_j^\prime:=\left\{ \underline{a}_j, \overline{a}_j \right\}$, the memory-one strategy of player $j$
\begin{align}
 T_j\left( a_j | \bm{a}^\prime \right) &= p\left( \bm{a}^\prime \right) \delta_{a_j, \underline{a}_j} + \left( 1-p\left( \bm{a}^\prime \right) \right) \delta_{a_j, \overline{a}_j}
 \label{eq:transition_two-point}
\end{align}
with
\begin{align}
 p\left( \underline{a}_j, a^\prime_{-j} \right) &:= \frac{1}{W} B\left( \underline{a}_j, a_{-j} \right) + 1 \quad (\forall a_{-j} \in A_{-j}) \nonumber \\
 p\left( \overline{a}_j, a^\prime_{-j} \right) &:= \frac{1}{W} B\left( \overline{a}_j, a_{-j} \right) \quad (\forall a_{-j} \in A_{-j})
 \label{eq:p_two-point}
\end{align}
is a ZD strategy unilaterally enforcing $\left\langle B \right\rangle^*=0$.
\end{proposition}
Indeed, we can find
\begin{align}
 \hat{T}_j\left( \underline{a}_j | \bm{a}^{\prime} \right) &= \frac{1}{W} B \left( \bm{a}^{\prime} \right) \quad \left( \forall a^\prime_j \in A_j^\prime, \forall a^{\prime}_{-j} \in A_{-j} \right).
\end{align}
They called such ZD strategies \emph{two-point autocratic strategies}.

Interestingly, the condition (\ref{eq:existence_condition_MH}) is exactly the same as Eq. (\ref{eq:condition_exsitence}), which is a necessary and sufficient condition for the existence of ZD strategies when an action space of each player is a finite set \cite{Ued2022b}.
Because two-point autocratic strategies are ZD strategies with $A_j$ consisting of two actions, they satisfy the necessary condition in Ref. \cite{Ued2022b}.
Therefore, the condition (\ref{eq:existence_condition_MH}) is also a necessary and sufficient condition for the existence of two-point autocratic strategies when all action spaces are finite sets.
These facts are summarized to the following theorem.
\begin{theorem}
\label{thm:ZDS_equlvalence}
If $A_k$ is a finite set for all $k\in \mathcal{N}$, the following three conditions are equivalent.
\begin{enumerate}[(i)]
 \item A ZD strategy of player $j$ controlling $B$ exists.
 \item A two-point autocratic strategy of player $j$ controlling $B$ exists.
 \item There exist two different actions $\overline{a}_j, \underline{a}_j \in A_j$ of player $j$ such that
\begin{align}
 B \left( \overline{a}_j, a_{-j} \right) &\geq 0 \quad \left( \forall a_{-j} \right) \nonumber \\
 B \left( \underline{a}_j, a_{-j} \right) &\leq 0 \quad \left( \forall a_{-j} \right),
\end{align}
and $B$ is not identically zero.
\end{enumerate}
\end{theorem}

By definition, two-point autocratic strategies are one-dimensional.
However, it should be noted that a strategy (\ref{eq:transition_two-point}) is not a ZD strategy in the sense of Definition \ref{def:ZDS}, since Eq. (\ref{eq:ZDS}) holds not for $\forall \bm{a}^{\prime} \in \mathcal{A}$, but for $\forall a^\prime_j \in A_j^\prime, \forall a^{\prime}_{-j} \in A_{-j}$.
Theorem \ref{thm:physicalZDS} claims that, when a ZD strategy exists, it can be expressed by one-dimensional transition probability using the whole action space $A_j$.
In Ref. \cite{McAHau2016}, it was pointed out that, if we use two-point autocratic strategies, feasible payoff regions are restricted because a two-point autocratic strategy uses only two actions.
Our one-dimensional ZD strategies can be useful for bypassing the problem of shrinkage of feasible payoff regions.

\section{Concluding remarks}
\label{sec:conclusion}
In this paper, we proved two theorems (Theorems \ref{thm:physicalZDS} and \ref{thm:ZDS_potential}) which can be useful for implementation of ZD strategies in repeated games.
The first theorem shows that ZD strategies, if exist, can be realized by one-dimensional transition probability.
This result can be useful when an action space has one-dimensional structure.
The second theorem claims that ZD strategies controlling payoffs exist in two-player non-trivial potential games.
We explained this result in the asymmetric prisoner's dilemma game as an example.

Before ending this paper, we make two remarks.
The first remark is related to extension of Theorem \ref{thm:ZDS_potential} to the case $N>2$.
Since the proof of Theorem \ref{thm:ZDS_potential} relies on the fact $N=2$, we cannot straightforwardly extend its proof to $N>2$ case.
Existence of ZD strategies controlling payoffs in potential games may be a specific result in $N=2$ case.
Further investigation is needed to clarify whether extension of our result to $N>2$ is possible or not.

The second remark is related to the implication of Theorem \ref{thm:physicalZDS}.
The form of the transition probability (\ref{eq:strategy_physZD_comp}) is simple, and, as discussed in Section \ref{sec:physicalZDS}, each element depends only on the quantity at the previous action profile $\bm{a}^\prime$.
Although this property is similar to that of the trap model, the stationary probability distribution cannot be written as a simple form, because the total dynamics is determined by strategies of all players.
We would like to investigate mathematical properties of the transition probability (\ref{eq:strategy_physZD_comp}) in future.
Furthermore, specifying situations where one-dimensional implementation (\ref{eq:strategy_physZD_comp}) of a ZD strategy is necessary is an interesting future problem.

\section*{Acknowledgement}
This study was supported by JSPS KAKENHI Grant Number JP20K19884 and Inamori Research Grants.



\section*{References}

\bibliography{ZDSphys}

\end{document}